\documentclass[11pt]{article}

\usepackage[utf8]{inputenc}
\usepackage[T1]{fontenc}
\usepackage{geometry}
\usepackage{amsmath, amssymb, amsthm, mathtools}
\usepackage{bm}
\usepackage{booktabs}
\usepackage{url}
\usepackage{algorithm2e}
\usepackage{tikz}
\usetikzlibrary{arrows.meta,positioning,calc}

\newtheorem{theorem}{Theorem}
\newtheorem{lemma}[theorem]{Lemma}

\newtheorem{corollary}[theorem]{Corollary}

\newtheorem{definition}{Definition}

\newcommand{\st}{\mathrm{ST}}
\newcommand{\la}{\mathrm{LA}}

\title{PSPACE-Completeness of Multi-Agent Path Finding for Large Agents}
\author{Maichi Zhang\thanks{
Graduate School of Informatics, Kyoto University.
\protect\url{zhangmaichi@amp.i.kyoto-u.ac.jp}} \and 
Naoyuki Kamiyama\thanks{
Graduate School of Informatics, Kyoto University.
\protect\url{kamiyama@amp.i.kyoto-u.ac.jp}.
This work was supported by JSPS KAKENHI Grant Number JP24K14825.} \and
Kanae Yoshiwatari\thanks{
Graduate School of Informatics, Kyoto University.
\protect\url{yoshiwatari.kanae.7p@kyoto-u.ac.jp}}}
\date{}

\begin{document}
\maketitle

\begin{abstract}
Multi-Agent Path Finding for Large Agents (LA-MAPF) is a geometric variant of MAPF in which agents are modeled as disks and conflicts are determined by physical overlap in the underlying Euclidean workspace. The goal of LA-MAPF is to decide whether there exists a sequence of conflict-free transitions from a start configuration to a goal configuration. Agafonov and Yakovlev proved that LA-MAPF is NP-hard. In this paper, we strengthen their result by proving that LA-MAPF is PSPACE-complete via a polynomial-time reduction from {\sc Restricted Sliding Tokens}.
\end{abstract}

\section{Introduction}

Multi-Agent Path Finding (MAPF) asks whether a collection of agents can be routed on a graph from prescribed start vertices to prescribed goal vertices without conflicts. In the classical formulation, agents are treated as dimensionless points, and conflicts are combinatorial. While this abstraction has fostered a rich algorithmic literature, it bypasses a fundamental geometric difficulty inherent in robotic applications: real-world agents have a physical volume.

LA-MAPF (MAPF for Large Agents) \cite{DBLP:conf/aaai/LiSF0KK19} incorporates this physical reality by embedding the graph in the Euclidean plane and modeling every agent as a disk of radius $r > 0$. When an agent traverses an edge, its disk sweeps a continuous tube around the straight-line segment representing the edge. Consequently, an agent may collide not only with another agent arriving at the same vertex, but also with a stationary agent located at a different vertex that lies too close to the moving agent's swept trajectory. This additional continuous-space conflict is a primary source of hardness in LA-MAPF.

While the continuous-space conflicts of LA-MAPF introduce significant geometric challenges, the computational complexity of the classical MAPF, where agents are treated as body-less points, has been extensively studied and depends heavily on the graph's topology and the specific task. The classical MAPF can be solved in polynomial time on undirected graphs \cite{kornhauser1984coordinating}. However, the classical MAPF becomes NP-complete on general directed graphs \cite{nebel2020computational}, though it can be solved in polynomial time for strongly connected directed graphs \cite{nebel2024computational}. The optimization variant of the classical MAPF, which seeks to minimize the makespan, is fundamentally intractable \cite{surynek2010optimization}. Finding optimal solutions like the sum of costs or makespan is NP-hard for undirected planar graphs \cite{yu2015intractability} and directed acyclic graphs \cite{tan2023intractability}, and standard 2D grid graphs \cite{banfi2017intractability}.

Recently, Agafonov and Yakovlev~\cite{DBLP:conf/ecai/AgafonovY25} investigated the computational complexity of LA-MAPF under the standard transition rule where exactly one agent moves along an edge at a time while all others remain stationary and proved that LA-MAPF is NP-hard by a reduction from 3-SAT. Their work established the fundamental intractability of the problem, but left its tight complexity class open. This is because it is not difficult to prove that this problem is in PSPACE (see Section 3).

\paragraph{Our Contribution.} In this paper, we settle the tight computational complexity of LA-MAPF.

\begin{theorem}[Main Theorem]
{\rm LA-MAPF} is {\rm PSPACE-complete}. 
\end{theorem}

Our proof demonstrates that the problem remains PSPACE-complete even under strict geometric constraints. This strict geometric setting motivates the specific orthogonal gadget designs utilized in our reduction.

\begin{corollary}
{\rm LA-MAPF} is {\rm PSPACE-complete} even when the underlying graph is a plane graph, and all agents' movements are restricted to horizontal and vertical directions.
\end{corollary}

The proof consists of two parts. First, we establish membership in PSPACE by demonstrating that reachability within its exponentially large state space can be decided using only polynomial memory via Savitch's Theorem~\cite{Savitch70}. Second, we prove the PSPACE-hardness of LA-MAPF by a polynomial-time reduction from a restricted version of {\sc Sliding Tokens}~\cite{BonsmaCereceda09, HearnDemaine05}. The outline of our proof is as follows: the independent set constraints of the graph can be flawlessly simulated by sparse one-hole wires rather than by fully packed corridors. Specifically, in our reduction, we contract only the token triangles into macro-vertices and utilize the orthogonal graph drawing algorithm \cite{BiedlK98} to lay out the macro-graph without edge crossovers. Token triangles become local constant-size gadgets, while token edges and link edges are transformed into orthogonal grid wires. By placing at most three agents on an orthogonal grid wire which represents a token edge or link edge and consists of two endpoints and at most two corner vertices, we physically simulate the independent set constraint of {\sc Sliding Tokens} using a single ``hole''.

\section{Preliminaries}

All graphs in this paper are finite and undirected. If a graph $G=(V,E)$ is embedded in the Euclidean plane, then we write $p(v)\in\mathbb{Q}^2$ for the rational coordinates of the point representing $v\in V$, where $\mathbb{Q}$ is the set of rational numbers. Distances are Euclidean, and let $\|x\|$ denote the distance between the origin and $x$ for each point $x \in \mathbb{Q}^2$. For two points $x,y\in\mathbb{Q}^2$, $\overline{xy}$ denotes the closed line segment joining them, and $\|x - y\|$ represents the distance between them. To strictly avoid irrational number computations in our proofs, all input coordinates and radii are given as rational numbers, and all distance comparisons are performed using squared rational distances.

\subsection{LA-MAPF}

In this subsection, we formally define LA-MAPF and its associated terminology.

\begin{definition}[LA-MAPF Instance]
An instance of \emph{LA-MAPF} is a tuple $I=(G,p,A,r,C_s,C_g)$, where $G=(V,E)$ is a graph, and $p:V\to\mathbb{Q}^2$ is an embedding of its vertices in the plane. Here, each edge $e=\{u,v\}\in E$ corresponds to the straight closed line segment connecting the points $p(u)$ and $p(v)$. Furthermore, $A=\{a_1, a_2, \ldots,a_k\}$ is a set of agents, $r>0$ is their common disk radius, and $C_s,C_g:A\to V$ are the start and goal configurations. In general, a configuration is an injective mapping $C:A\to V$ assigning each agent to a vertex.
\end{definition}

In LA-MAPF, a conflict between agents can occur either at the vertices or along the edges during continuous motion. A configuration is \emph{vertex-feasible} (i.e., devoid of vertex conflicts) if the physical disks of every two distinct agents $a_i, a_j \in A$ do not overlap:
\[
  \|p(C(a_i)) - p(C(a_j))\|^2 \ge 4r^2.
\]
We assume that the start configuration $C_s$ and the goal configuration $C_g$ are vertex-feasible.

For a vertex-feasible configuration $C$ and another configuration $C'$, we write $C \rightsquigarrow C'$ if $C'$ is vertex-feasible and there exists some agent $a_i \in A$ satisfying the following conditions:
\begin{itemize}
    \item Assume that $u = C(a_i)$ and $v = C'(a_i)$. Then $u \neq v$ and $\{u, v\} \in E$.
    
    \item $C(a_j) = C'(a_j)$ for every agent $a_j \in A \setminus \{a_i\}$.
    
    \item For every agent $a_j \in A \setminus \{a_i\}$, we have
    $$ \mathrm{dist}(p(C(a_j)), \overline{p(u)p(v)})^2 \geq 4r^2, $$
    where, for a point $x \in \mathbb{Q}^2$ and a closed line segment $\overline{yz}$ such that $y, z \in \mathbb{Q}^2$, $\mathrm{dist}(x, \overline{yz})$ is defined as the minimum Euclidean distance between $x$ and a point on the segment, i.e., $\mathrm{dist}(x, \overline{yz}) = \min_{\lambda \in [0,1]} \|x - ((1 - \lambda)y + \lambda z)\|$.
\end{itemize}
 In other words, $C \rightsquigarrow C'$ means that the \textit{transition} from the vertex-feasible configuration $C$ to the configuration $C'$ is \textit{collision-free}, i.e., devoid of vertex and edge conflicts. Under the condition that only one agent $a_i$ moves from $u = C(a_i)$ to $v = C'(a_i)$ along the edge $\{u, v\} \in E$, while all other agents remain stationary at their current vertices. Furthermore, for every stationary agent $a_j$ at the vertex $C(a_j) = C'(a_j)$, its distance to the moving agent's swept segment is geometrically safe. A \textit{conflict-free plan} from a vertex-feasible configuration $C$ to another configuration $C'$ is a finite sequence of configurations $C_1, C_2, \dots, C_\ell$ such that $C_1 = C$, $C_\ell = C'$, $C_i \rightsquigarrow C_{i+1}$ for every integer $i \in \{1, 2, \dots, \ell - 1\}$.

We are now ready to define our problem.
\begin{definition}[LA-MAPF]
Given an instance $I=(G,p,A,r,C_s,C_g)$ of \emph{LA-MAPF}, the goal is to determine whether there exists a conflict-free plan from $C_s$ to $C_g$.
\end{definition}

\subsection{{\sc Sliding Tokens}}

In this subsection, we define {\sc Restricted {\sc Sliding Tokens}}, which serves as the source problem for our PSPACE-completeness reduction.

In {\sc Sliding Tokens}, we are given a graph $G_{\rm ST}$. A \emph{token configuration} of $G_{\rm ST}$ is defined as an independent set in $G_{\rm ST}$ (i.e., a subset of vertices of $G_{\rm ST}$ such that there does not exist an edge of $G_{\rm ST}$ connecting any pair of vertices in this subset). A token configuration represents the places where tokens reside. Given a start token configuration $T_s$ and a goal token configuration $T_g$, the problem asks whether $T_s$ can be transformed into $T_g$ by a sequence of valid token moves. A \emph{valid token move} slides a single token along an edge of $G_{\rm ST}$ to an adjacent vertex, provided that the resulting set of occupied vertices remains an independent set, i.e., a token configuration. An instance of {\sc Sliding Tokens} is denoted as $(G_{\rm ST}, T_s, T_g)$.

We reduce from the following restricted version of {\sc Sliding Tokens}, which is called \emph{{\sc Restricted {\sc Sliding Tokens}}} (Figure \ref{fig:rst-example}). It is known that {\sc Restricted {\sc Sliding Tokens}} is PSPACE-complete \cite{BonsmaCereceda09,HearnDemaine05}.
 
\begin{definition}[{\sc Restricted {\sc Sliding Tokens}} Instance]
An instance of \emph{{\sc Restricted {\sc Sliding Tokens}}} is an instance $(G_{\rm ST},T_s,T_g)$ of {\sc Sliding Tokens} satisfying the following properties:

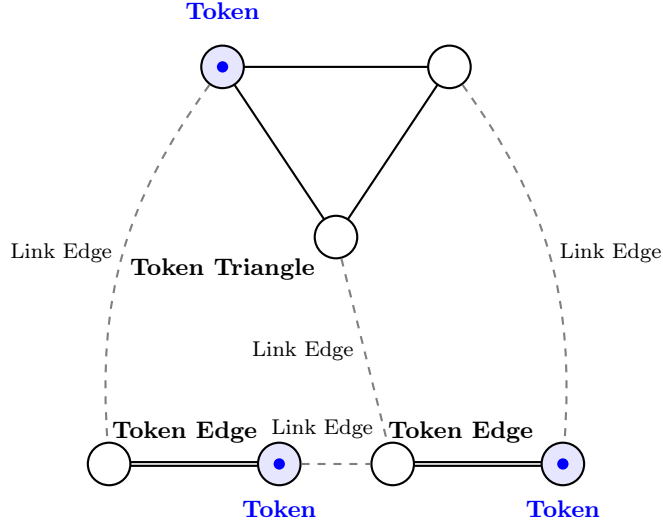
\begin{figure}[htbp]
\centering
\begin{tikzpicture}[
    scale=1.5,
    every node/.style={font=\small},
    tokenbg/.style={circle, draw=black, fill=blue!10, thick, inner sep=2pt, minimum size=16pt},
    empty/.style={circle, draw=black, fill=white, thick, inner sep=2pt, minimum size=16pt},
    token/.style={circle, fill=blue, inner sep=1.5pt}
]

\coordinate (X1) at (0, 2);
\coordinate (X2) at (2, 2);
\coordinate (X3) at (1, 0.5);

\draw[thick] (X1) -- (X2) -- (X3) -- cycle;

\node[tokenbg] (nX1) at (X1) {}; 
\node[token] at (X1) {};
\node[empty] (nX2) at (X2) {};
\node[empty,] (nX3) at (X3) {};

\node[above=0.2cm of nX1, text=blue, font=\footnotesize\bfseries] {Token};
\node[above=-0.7cm of X3, xshift=-1.5cm, font=\footnotesize\bfseries] {Token Triangle};

\coordinate (Y1) at (-1, -1.5);
\coordinate (Y2) at (0.5, -1.5);

\draw[thick, double, distance=2pt] (Y1) -- (Y2);

\node[empty] (nY1) at (Y1) {};
\node[tokenbg] (nY2) at (Y2) {}; 
\node[token] at (Y2) {};

\node[below=0.05cm of nY2, text=blue, font=\footnotesize\bfseries] {Token};
\node[above=0.15cm of Y1, xshift=1.0cm, font=\footnotesize\bfseries] {Token Edge};

\coordinate (Z1) at (1.5, -1.5);
\coordinate (Z2) at (3, -1.5);

\draw[thick, double, distance=2pt] (Z1) -- (Z2);

\node[empty] (nZ1) at (Z1) {};
\node[tokenbg] (nZ2) at (Z2) {}; 
\node[token] at (Z2) {};

\node[below=0.05cm of nZ2, text=blue, font=\footnotesize\bfseries] {Token};
\node[above=0.15cm of Z1, xshift=0.9cm, font=\footnotesize\bfseries] {Token Edge};

\draw[dashed, thick, gray] (nX1) edge[bend right=20] node[left, font=\scriptsize, text=black] {Link Edge} (nY1);
\draw[dashed, thick, gray] (nX2) edge[bend left=20] node[right, font=\scriptsize, text=black] {Link Edge} (nZ2);
\draw[dashed, thick, gray] (nX3) -- node[left, font=\scriptsize, text=black] {Link Edge} (nZ1);
\draw[dashed, thick, gray] (nY2) -- node[pos=0.3, yshift=0.2cm, font=\scriptsize, text=black, anchor=south] {Link Edge} (nZ1);

\end{tikzpicture}
\caption{An example of an instance $(G_{\rm ST},T_s,T_g)$ of {\sc Restricted {\sc Sliding Tokens}}. The graph is partitioned into token triangles (solid lines) and token edges (double lines), interconnected by link edges (dashed lines). In a standard token configuration, exactly one token (blue dot) resides in each token triangle and each token edge.}
\label{fig:rst-example}
\end{figure}

\begin{enumerate}
  \item The graph $G_{\rm ST}$ has maximum degree at most three, and its vertex set is partitioned into pairwise disjoint \emph{token triangles}, which are isomorphic to $K_3$ and \emph{token edges}, which are isomorphic to $K_2$. The degree of every vertex in a token triangle is exactly three.
  \item The edges of $G_{\rm ST}$ not contained in any token triangle or token edge are called \emph{link edges}. A link edge joins a vertex of one token triangle or token edge to a vertex of another token triangle or token edge. 
  \item The graph $G_{\rm ST}$ has a planar embedding where every token triangle bounds a face.

  \item The start and goal token configurations $T_s,T_g$ are \emph{standard}: each token triangle and each token edge contains exactly one token.
\end{enumerate}
\end{definition}

A fundamental property of these restricted instances is that every reachable token configuration remains standard. A token can never traverse a link edge; the first token to enter another token triangle or token edge would immediately violate the independent set constraint by becoming adjacent to the token already residing there. Therefore, the sole topological function of a link edge is to impose a binary exclusion constraint: its two endpoint vertices cannot be occupied simultaneously.

\begin{definition}[{\sc Restricted {\sc Sliding Tokens}}]
Given an instance $(G_{\rm ST},T_s,T_g)$ of~\textup{\textsc{Restricted Sliding Tokens}}, the goal is to determine whether there exists a sequence of valid token moves that transforms the start token configuration $T_s$ into the goal token configuration $T_g$.
\end{definition}

\section{Membership in PSPACE}

\begin{lemma}
$\mathrm{LA\text{-}MAPF}$ belongs to $\mathrm{PSPACE}$.
\end{lemma}

\begin{proof}
By Savitch's Theorem ($\text{NPSPACE} = \text{PSPACE}$) \cite{Savitch70}, it suffices to show that LA-MAPF can be solved by a non-deterministic algorithm using polynomial space.

Let $n$ be the binary encoding length of the input instance $I = (G, p, A, r, C_s, C_g)$. Note that the number of vertices $|V|$ and the number of agents $k$ are both bounded by $n$. At any time step, a configuration $C$ assigns $k$ agents to $k$ vertices. Representing this assignment requires $k \lceil \log_2 |V| \rceil = O(n \log n)$ bits. The total number of distinct configurations is bounded by $|V|^k$. 

We construct a non-deterministic algorithm that maintains only the current configuration $C_{\text{curr}}$, guesses a valid next configuration $C_{\text{next}}$, and maintains a step counter $S$ to prevent infinite loops. If $S > |V|^k$, the system has exceeded the maximum number of possible configurations and must be trapped in a loop, so the algorithm rejects. The binary counter $S$ requires at most $\log_2(|V|^k) + 1= k \log_2 |V| + 1= O(n \log n)$ bits of storage.
In each step, the algorithm non-deterministically guesses $C_{\text{next}}$ and deterministically verifies that the transition is collision-free.

At any given point, the algorithm stores at most two configurations, the step counter, and the polynomial-size memory required for the geometric verification. Since the maximum memory space is bounded by a polynomial in $n$, LA-MAPF is in NPSPACE. By Savitch's Theorem, LA-MAPF belongs to PSPACE.
\end{proof}

\section{PSPACE-Hardness}
In this section, we prove that LA-MAPF is PSPACE-hard by a reduction from {\sc Restricted Sliding Tokens}. Let $(G_{\rm ST}, T_s, T_g)$ be a given instance of {\sc Restricted Sliding Tokens}. We then construct, in polynomial time, an instance $I_{\la}$ of LA-MAPF from the instance $(G_{\rm ST},T_s,T_g)$ of {\sc Restricted {\sc Sliding Tokens}} such that
\[
  (G_{\st},T_s,T_g) \text{ is a yes-instance if and only if } I_{\la} \text{ is a yes-instance}.
\]

\subsection{Macro graph and orthogonal grid embedding}

To map the instance $(G_{\st},T_s,T_g)$ of {\sc Restricted {\sc Sliding Tokens}} on the physical workspace of LA-MAPF, we first construct a macro-graph $H=(V_H, E_H)$ from $G_{\st}$. The crucial step in our reduction is that we contract only the token triangles into macro-vertices, while preserving the token edges as topological edges.

Before detailing the ingredients of $H$, we define a \emph{port} and a \emph{component} in $G_{\st}$. Let $V(G_{\st})$ denote the vertex set of $G_{\st}$. Because the vertices of $G_{\st}$ are partitioned into token triangles and token edges, we collectively refer to each individual token triangle and token edge as a \emph{component} of $G_{\st}$. Consequently, every vertex $v \in V(G_{\st})$ belongs to exactly one such component. We formally refer to each original vertex $v \in V(G_{\st})$ as a \emph{port} of its respective component.

Formally, the ingredients of $H$ are defined as follows:
\begin{itemize}
    \item \textbf{Vertices ($V_H$)}: For each token triangle in $G_{\st}$, $V_H$ contains exactly one representing vertex, called a \emph{macro-vertex}. For every token edge in $G_{\st}$, its two ports are directly included in $V_H$ as distinct vertices, which we define as \emph{port vertices} of this  token edge.
    \item \textbf{The edge set ($E_H$)}: The edge set $E_H$ consists of all the original token edges (connecting their two port vertices) and all the link edges (connecting either macro-vertices or port vertices, exactly as in $G_{\st}$).
\end{itemize}

This contraction preserves planarity and bounds the maximum degree. A macro-vertex representing a token triangle connects to exactly three link edges, fixing its degree at exactly three. Meanwhile, a port vertex of a token edge connects to exactly one token edge and either one or two link edges, ensuring its degree is between two or three. Thus, $H$ is a planar graph of maximum degree at most three.

Using the algorithm for orthogonal graph drawing \cite{BiedlK98}, $H$ admits a planar orthogonal embedding on a $|V_H| \times |V_H|$ grid. Furthermore, we can compute such an embedding in linear time. In this embedding, every vertex in $V_H$ is placed at an integer grid point, and every edge in $E_H$ is routed as a sequence of orthogonal (horizontal and vertical) straight-line segments. We refer to these orthogonally routed edges as \emph{wires}. The algorithm guarantees that each wire has at most two bends. We uniformly scale this grid by a large factor: $\Lambda = 100r$.
This scaling guarantees that all non-incident gadgets and orthogonal wire segments are separated by distances vastly exceeding $2r$, preventing unintended edge conflicts. 

Once the macro-graph is orthogonally embedded and scaled, we geometrically truncate (or ``crop'') the extreme ends of every orthogonally routed link edge wire. This structural cropping is performed after the embedding to create the necessary physical space to interface these link edge wires with our constant-size local gadgets (token triangle gadgets and token edge gadgets) without causing structural intersections (this cropped interface is visually demonstrated for individual gadgets in Figures~\ref{fig:triangle-gadget} and \ref{fig:port-gadget}, and for a complete link edge wire in Figure~\ref{fig:link-wire}).

\subsection{Token triangle gadgets}

For each token triangle $X$ of $G_{\rm ST}$, we assume that the corresponding macro-vertex in $H$ is embedded at $(x_X, y_X) \in \mathbb{Q}^2$. We assign exactly one \emph{v-agent} to this gadget.

By our earlier definition, the token triangle $X$ consists of exactly three ports. Let us index these three ports as 1, 2, and 3. To enforce orthogonal continuous motion, we define the placement of four port vertices $p_{X,1}, p_{X,2}, p_{X,3}, p_{X,t}$ symmetrically around the grid axes. They form a rectangle graph as follows (as illustrated in Figure \ref{fig:triangle-gadget}):

\[
  p_{X,1} = (x_X - 0.5r, y_X), \qquad p_{X,2} = (x_X - 0.5r, y_X + r),
\]
\[
  p_{X,3} = (x_X + 0.5r, y_X), \qquad p_{X,t} = (x_X + 0.5r, y_X + r).
\]

\begin{figure}[h]
\centering
\begin{tikzpicture}[scale=1.2, every node/.style={font=\small, align=center}]

  \draw[loosely dashed, gray!40, thick] (2.25, 0) -- (4.0, 0) node[right, text=gray] {$y = y_X$};
  \draw[loosely dashed, gray!40, thick] (0, -1.0) -- (0, 4.5) node[above, text=gray] {$x = x_X$};

  \fill[black] (0,0) circle (1.5pt) node[below=4pt, font=\scriptsize] {Anchor};

  \coordinate (P1) at (-0.5, 0);
  \coordinate (P2) at (-0.5, 1);
  \coordinate (P3) at (0.5, 0);
  \coordinate (Pt) at (0.5, 1);

  \coordinate (B1) at (-2.25, 0);
  \coordinate (B2) at (0, 2.5);
  \coordinate (B3) at (2.25, 0);

  \draw[black, thick, ->] (B1) -- (-3.5, 0) node[above,text=black] {Link Edge $\ell_1$};
  \draw[black, thick, ->] (B2) -- (0, 4.0) node[right,text=black] {Link Edge $\ell_2$};
  \draw[black, thick, ->] (B3) -- (3.5, 0) node[above,text=black] {Link Edge $\ell_3$};

  \draw[dashed, red, ultra thick] (P2)--(B2);
  \draw[dashed, red, ultra thick] (Pt)--(B2);
  \node[red, align=center] at (-1.5, 2) {Overlap\\($d^2 = 2.5r^2 < 4r^2$)};
  
  \draw[dashed, red, ultra thick] (P3)--(B3);

  \fill[blue!15] (P1) circle (1);
  \draw[blue, thick] (P1) circle (1);
  \draw[blue, ->, >=stealth] (P1) -- (-1.207, 0.707) node[midway, above right, inner sep=1pt] {$r$};
  \node[blue, font=\bfseries] at (0.8, -0.8) {$v$-agent};

  \draw[thick] (P1)--(P2)--(Pt)--(P3)--cycle;
  \draw[dashed, red, ultra thick] (P1)--(B1) node[left, below=2pt, text=black] {};

  \fill [blue] (P1) circle (2pt) node[below left, xshift=2pt, yshift=-2pt] {$p_{X,1}$};
  \draw[thick, fill=white] (P2) circle (2pt) node[above left, xshift=2pt] {$p_{X,2}$};
  \draw[thick, fill=white] (P3) circle (2pt) node[below right, xshift=-2pt, yshift=-2pt] {$p_{X,3}$};
  \draw[thick, fill=white] (Pt) circle (2pt) node[above right, xshift=-2pt] {$p_{X,t}$};

  \draw[thick, fill=white] (B1) circle (2.5pt) node[above, yshift=2pt] {$B_{1,\ell_1}^{X}$};
  \fill[red!70!white] (B2) circle (2.5pt) node[right, xshift=2pt] {$B_{2,\ell_2}^{X}$\\\text{}};
  \fill[red!70!white] (B3) circle (2.5pt) node[above, yshift=-4pt] {$B_{3,\ell_3}^{X}$\\\text{}};

\end{tikzpicture}
\caption{A token triangle gadget restructured into an orthogonal rectangle. The gadget is aligned with the intersecting macro-graph grid lines. The addition of $p_{X,t}$ forms a rectangle graph, avoiding diagonal continuous motions. In the depicted configuration, the vertices $B_{2,\ell_2}^{X}$ and $B_{3,\ell_3}^{X}$ are occupied by other agents. Due to geometric conflicts, the $v$-agent is physically restricted to $p_{X,1}$ and cannot move to the vertices $p_{X,2}, p_{X,t}$, or $p_{X,3}$.}
\label{fig:triangle-gadget}
\end{figure}
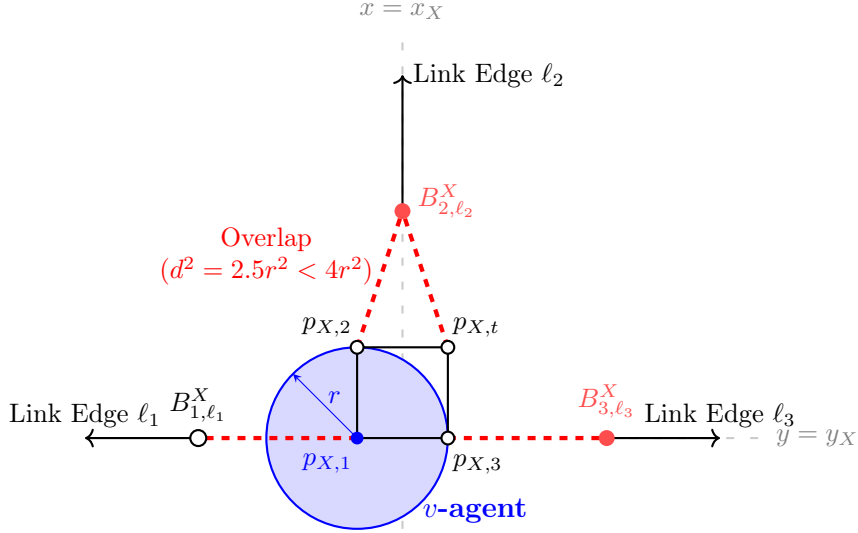

We map the port 1 to $p_{X,1}$, the port 3 to $p_{X,3}$, and the port 2 to $\{p_{X,2}, p_{X,t}\}$. 
Assume that the token triangle $X$ is incident to three link edges $\ell_1, \ell_2$, and $\ell_3$ in $G_{\rm ST}$. For these incident link edges routed orthogonally by the grid embedding, we introduce three external \emph{blocking vertices}, denoted as $B^X_{1,\ell_1}, B^X_{2,\ell_2}$, and $B^X_{3,\ell_3}$. Here, $B^X_{i,\ell_k}$ represents the blocking vertex corresponding to the port $i$ of $X$ and the incident link edge $\ell_k$. We place these blocking vertices as follows:
\[
  B_{1,\ell_1}^{X} = (x_X - 2.25r, y_X), \qquad B_{2,\ell_2}^{X} = (x_X, y_X + 2.5r), \qquad B_{3,\ell_3}^{X} = (x_X + 2.25r, y_X).
\]
This geometric arrangement can be rotated around the anchor $(x_X, y_X)$ by multiples of $90^\circ$ to accommodate any arbitrary combination of up to three orthogonal routing directions chosen by the macro-graph embedding. 

The squared distance checks yield:
\[
  \|B_{1,\ell_1}^{X} - p_{X,1}\|^2 = (-1.75r)^2 = 3.0625r^2 < 4r^2,
\]
\[
  \|B_{2,\ell_2}^{X} - p_{X,2}\|^2 = \|B_{2,\ell_2}^{X} - p_{X,t}\|^2 = (\pm 0.5r)^2 + (1.5r)^2 = 2.5r^2 < 4r^2,
\]
\[
  \|B_{3,\ell_3}^{X} - p_{X,3}\|^2 = (1.75r)^2 = 3.0625r^2 < 4r^2.
\]

Crucially, non-target interactions are rigorously safe. For instance, the squared distance from $B_{1,\ell_1}^{X}$ to $p_{X,2}$ is $(-1.75r)^2 + (r)^2 = 4.0625r^2 > 4r^2$, guaranteeing sufficient geometric separation. Consequently, when the $v$-agent occupies $p_{X,1}$, its physical disk would overlap with that of any agent at $B_{1,\ell_1}^{X}$, preventing $B_{1,\ell_1}^{X}$ from being occupied. Meanwhile, it imposes no geometric restrictions on $B_{2,\ell_2}^{X}$ and $B_{3,\ell_3}^{X}$, leaving them fully accessible. For the port 2, a geometric conflict exists between $B_{2,\ell_2}^{X}$ and the vertices $p_{X,2}$ and $p_{X,t}$. This mutual exclusion ensures that the $v$-agent cannot move to $p_{X,2}$ or $p_{X,t}$ if another agent is already stationed at $B_{2,\ell_2}^{X}$; conversely, once the $v$-agent occupies either of $p_{X,2}$ or $p_{X,t}$, no other agent can occupy $B_{2,\ell_2}^{X}$.

\subsection{Token edge gadgets}

A token edge $Y$ of $G_{\rm ST}$ connecting the ports $1$ and $2$ is physically routed along its assigned orthogonal wire in the macro-grid $H$. 

Let $p_{Y,1}$ and $p_{Y,2}$ be the port vertices in $V_H$ corresponding to the two ports $1$ and $2$ of the token edge $Y$. Since the orthogonal embedding algorithm~\cite{BiedlK98} guarantees that each routed edge has at most two bends, the orthogonal wire contains at most two grid corners, meaning the number of vertices is bounded by a small constant (at most four vertices: the two port vertices and at most two orthogonal bends).

To simulate a token move on $Y$, we assign a fixed number of $v$-agents to the token edge wire, exactly one fewer than the total number of vertices (hence, at most three $v$-agents). These agents reside exclusively on the vertices. Because there is exactly one more vertex than the number of agents, there remains exactly one unoccupied vertex at any time, which serves as the \emph{hole}. Due to the physical footprint of the agents, they cannot bypass one another within the narrow geometric corridors. Therefore, simulating a macro-level token slide does not involve a single $v$-agent traveling the entire wire. Instead, it is executed as a sequence of a few transitions (at most three steps): an adjacent $v$-agent continuously slides along the straight-line segment into the hole, leaving its previous vertex as the new hole. For instance, consider a wire with four vertices sequentially ordered as $p_{Y,1}$, $\overline{p}_{Y,1}$, $\overline{p}_{Y,2}$, and $p_{Y,2}$ (as depicted in Figure \ref{fig:port-gadget}). Suppose the token moves from the port $1$ to the port $2$ on $Y$. To represent this locally, the $v$-agents sequentially transition one by one to propagate the hole toward $p_{Y,1}$. Specifically, the hole is initially at $p_{Y,2}$, and the agent at $\overline{p}_{Y,2}$ slides to $p_{Y,2}$, followed by the agent at $\overline{p}_{Y,1}$ sliding to $\overline{p}_{Y,2}$, and finally the agent at $p_{Y,1}$ sliding to $\overline{p}_{Y,1}$. These transitions cascades the hole all the way to $p_{Y,1}$, rigorously forcing the contiguous chain of $v$-agents to occupy all other vertices on the wire, definitively stationing a $v$-agent at $p_{Y,2}$. This configuration logically represents that the port $2$ is occupied by the token in $G_{\st}$.

At $p_{Y,1}$, exactly one orthogonal direction is consumed by the token edge wire. The remaining up to three available orthogonal directions are utilized for the incident link edge wires routed by the orthogonal embedding. For the wire corresponding to each link edge $\ell_m$ approaching from a specific orthogonal direction, we place a blocking vertex $B_{1,\ell_m}^{Y}$ exactly on that directional axis at distance $1.5r$ from $p_{Y,1}$:
\[
  \|p_{Y,1} - B_{1,\ell_m}^{Y}\|^2 = (1.5r)^2 = 2.25r^2 < 4r^2.
\]
Notice that there is no edge between $p_{Y,1}$ and $B_{1,\ell_m}^{Y}$. The constraint is geometric: if a $v$-agent is at $p_{Y,1}$, its physical disk overlaps $B_{1,\ell_m}^{Y}$, ensuring it remains geometrically blocked. Conversely, if the hole is at $p_{Y,1}$, no $v$-agent is present there. Let $\ell_n$ be a link edge incident to the port 1, and $B^Y_{1,\ell_n}$ denote its corresponding blocking vertex. Because $p_{Y,1}$ is unoccupied, $B^Y_{1,\ell_n}$ is geometrically safe to be occupied by agents.

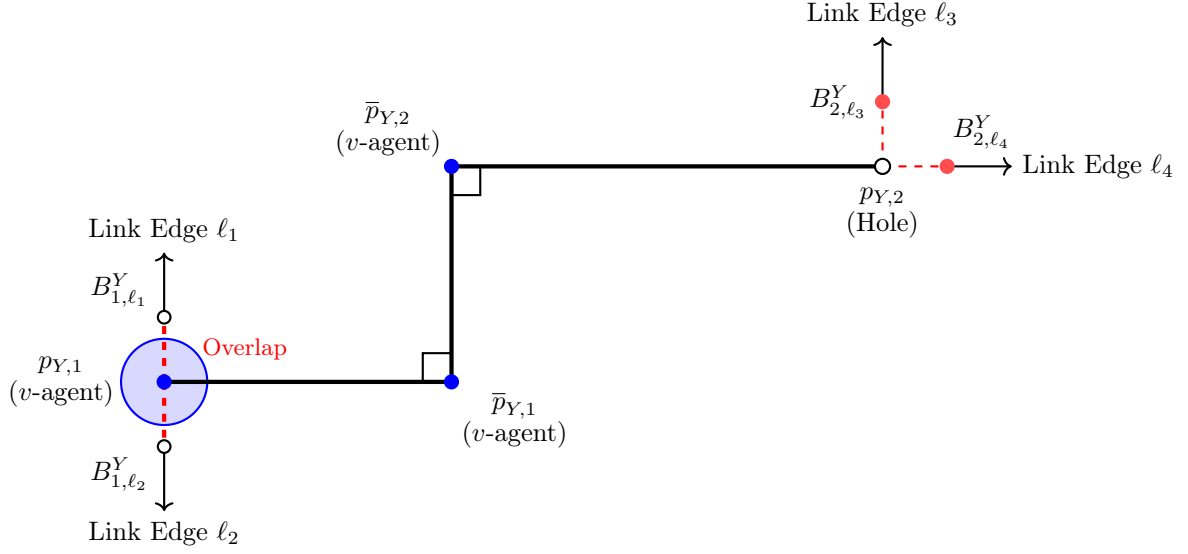
\begin{figure}[h]
\centering
\begin{tikzpicture}[scale=0.95, every node/.style={font=\small, align=center}]

  \coordinate (Tu) at (0,0);
  \coordinate (Tv) at (10,3);

  \coordinate (C1) at (4,0);
  \coordinate (C2) at (4,3);

  \coordinate (Bu1) at (0, 0.9);
  \coordinate (Bu2) at (0, -0.9);

  \coordinate (Bv1) at (10, 3.9);
  \coordinate (Bv2) at (10.9, 3);

  \draw[black, thick, ->] (Bu1) -- (0, 1.8) node[above, text=black] {Link Edge $\ell_1$};
  \draw[black, thick, ->] (Bu2) -- (0, -1.8) node[below, text=black] {Link Edge $\ell_2$};
  \draw[black, thick, ->] (Bv1) -- (10, 4.8) node[above, text=black] {Link Edge $\ell_3$};
  \draw[black, thick, ->] (Bv2) -- (11.8, 3) node[right, text=black] {Link Edge $\ell_4$};

  \fill[blue!15] (Tu) circle (0.6); 
  \draw[blue, thick] (Tu) circle (0.6);

  \draw[ultra thick, black] (Tu) -- (C1) -- (C2) -- (Tv);

  \node[above, text=black] at (6, 3) { };

  \draw[thick, black] (3.6,0) -- (3.6,0.4) -- (4,0.4);       
  \draw[thick, black] (4.4,3) -- (4.4,2.6) -- (4,2.6);       

  \draw[dashed, red, ultra thick] (Tu) -- (Bu1) node[midway, right, xshift=10pt, font=\footnotesize] {Overlap};
  \draw[dashed, red, ultra thick] (Tu) -- (Bu2) node[midway, right, xshift=10pt, font=\footnotesize] { };
  
  \fill[blue] (Tu) circle (3pt) node[left, xshift=-15pt, text=black] {$p_{Y,1}$\\($v$-agent)};
  \fill[blue] (C1) circle (3pt) node[below right, text=black] {$\overline{p}_{Y,1}$\\($v$-agent)};
  \fill[blue] (C2) circle (3pt) node[above left, text=black] {$\overline{p}_{Y,2}$\\($v$-agent)};
  
  \draw[dashed, red, thick] (Tv) -- (Bv1) node[midway, right, xshift=2pt, font=\footnotesize] { };
  \draw[dashed, red, thick] (Tv) -- (Bv2);

  \draw[thick, black, fill=white] (Tv) circle (3pt) node[below, yshift=-4pt, text=black] {$p_{Y,2}$\\(Hole)};

  \draw[thick, fill=white] (Bu1) circle (2.5pt) node[above left, xshift=-2pt] {$B_{1,\ell_1}^{Y}$};
  \draw[thick, fill=white] (Bu2) circle (2.5pt) node[below left, xshift=-2pt] {$B_{1,\ell_2}^{Y}$};
  
  \fill[red!70!white] (Bv1) circle (3pt) node[left, xshift=-2pt, text=black] {$B_{2,\ell_3}^{Y}$};
  \fill[red!70!white] (Bv2) circle (3pt) node[yshift=13pt, xshift=13pt, text=black] {$B_{2,\ell_4}^{Y}$};

\end{tikzpicture}
\caption{The wire contains at most four vertices, it requires at most three $v$-agents. Here, the hole is pushed to $p_{Y,1}$, forcing a $v$-agent (blue disk) to reside at $p_{Y,2}$. This geometric presence overlaps and blocks $B_{2,\ell_3}^{Y}$ and $B_{2,\ell_4}^{Y}$. Conversely, the absence of an agent at $p_{Y,1}$ allows $B_{1,\ell_1}^{Y}$ and $B_{1,\ell_2}^{Y}$ to safely host agents.}
\label{fig:port-gadget}
\end{figure}

\subsection{Link edge wires and constraint transmission}
Consider a link edge ${\ell_k}$ of $G_{\st}$ joining the port $i$ of one component $X$ to the port $j$ of another component $Y$. Recall that $B_{i,{\ell_k}}^{X}$ and $B_{j,{\ell_k}}^{Y}$ are the blocking vertices corresponding to the port vertex $p_{X,i}$ and the port vertex $p_{Y,j}$, respectively. In the scaled orthogonal macro-embedding, the route for ${\ell_k}$ is initially an orthogonal wire connecting the two respective grid points. To accommodate the physical footprint of the endpoint gadgets, this continuous wire is geometrically cropped at both ends. 

Utilizing the property that the orthogonal algorithm produces at most two bends per edge, the cropped orthogonal route for ${\ell_k}$ contains at most two intermediate grid corners. We define $W_{{\ell_k}}$ as the wire of vertices  along this remaining route, starting at $B_{i,{\ell_k}}^{X}$, passing through the intermediate corners, and ending at $B_{j,{\ell_k}}^{Y}$. Because there are at most two corners, $W_{\ell_k}$ contains at most four vertices. We assign \emph{blocker-agents} ($b$-agents) to $W_{\ell_k}$, providing exactly one fewer $b$-agent than the total number of vertices on the wire (hence, at most three $b$-agents). Thus, $W_{\ell_k}$ always maintains exactly one hole at all times. The mechanism of constraint transmission along this wire is illustrated in Figure \ref{fig:link-wire}.

The following geometric invariant forms the core of our reduction.

\begin{lemma}[Constraint Transmission Invariant]
Let ${\ell_k}$ be a link edge in $G_{\rm ST}$ connecting the port $i$ of a component X and the port $j$ of a component Y. In every configuration reachable from the start configuration via a conflict-free plan:
\begin{enumerate}
  \item[(1)] The wire $W_{\ell_k}$ contains exactly one hole.
  \item[(2)] If a $v$-agent occupies the port vertex $p_{X,i}$, then $B_{i,{\ell_k}}^{X}$ must be the hole of $W_{\ell_k}$, which forces $B_{j,{\ell_k}}^{Y}$ to be occupied by a $b$-agent.
  \item[(3)] If a $v$-agent occupies the port vertex $p_{X,i}$, then the port vertex $p_{Y,j}$ cannot be simultaneously occupied by another $v$-agent, preserving the independent set constraint enforced by $\ell_k$.
\end{enumerate}
\end{lemma}

\begin{proof}
\textbf{(1)} Since $W_{\ell_k}$ is an isolated wire component in the embedded graph and agents cannot jump across non-incident wires, the total number of $b$-agents on $W_{\ell_k}$ remains conserved. Because there is exactly one fewer agent than vertices, $W_{\ell_k}$ always contains exactly one hole.

\textbf{(2)} If a $v$-agent occupies the port vertex $p_{X,i}$, its physical disk overlaps $B_{i,\ell_k}^X$. This overlap dictates that no $b$-agent can reside at $B_{i,\ell_k}^X$, locking $B_{i,\ell_k}^X$ as the hole. By the result of (1), the wire $W_{\ell_k}$ contains exactly one hole. Since this unique hole is located at $B_{i,\ell_k}^X$, all other vertices on $W_{\ell_k}$ must be occupied by $b$-agents. Consequently, the opposite endpoint vertex $B_{j,\ell_k}^Y$ is occupied by a $b$-agent.

\textbf{(3)} Suppose a $v$-agent occupies the port vertex $p_{X,i}$. By the result of (2), the blocking vertex $B_{j,\ell_k}^Y$ must be occupied by a $b$-agent. If another $v$-agent attempts to move into the port vertex $p_{Y,j}$, its physical disk would geometrically conflict with the $b$-agent already stationed at $B_{j,\ell_k}^Y$. Thus, simultaneous occupation of the connected ports is impossible.
\end{proof}

\begin{figure}[h]
\centering
\begin{tikzpicture}[scale=1.1, every node/.style={font=\small, align=center}]

  \coordinate (PortA) at (0,1.5);
  \coordinate (BA) at (1.5,1.5);

  \coordinate (C1) at (4,1.5);   
  \coordinate (C2) at (4,-0.5);  
  \coordinate (BB) at (6.5,-0.5);
  
  \fill[red!15] (BB) circle (0.6); 
  \draw[red, thick] (BB) circle (0.6);
  \coordinate (PortB) at (8.0,-0.5);

  \draw[thick, gray] (BA) -- (C1) -- (C2) -- (BB);

  \draw[gray] (3.6,1.5) -- (3.6,1.1) -- (4,1.1);
  \draw[gray] (4.4,-0.5) -- (4.4,-0.1) -- (4,-0.1);

  \draw[dashed, red, ultra thick] (PortA) -- (BA) node[midway, above, yshift=2pt] {$d^2 < 4r^2$};
  \draw[dashed, red, ultra thick] (BB) -- (PortB) node[midway, above, yshift=2pt] {};
  
  \fill[blue!70!white] (PortA) circle (3pt) node[left, xshift=-4pt, text=black] {Port Vertex $p_{X,i}$\\($v$-agent)};
  \draw[thick, fill=white] (BA) circle (3pt) node[below, yshift=-4pt] {$B_{i,{\ell_k}}^{X}$\\ (Hole)};
  
  \fill[red!70!white] (C1) circle (3pt) node[below left, xshift=-10pt, yshift=-2pt, text=black] {${M}_{i,{\ell_k}}^{X}$\\($b$-agent)};
  \fill[red!70!white] (C2) circle (3pt) node[above left, xshift=-10pt, yshift=-6pt, text=black] {${M}_{j,{\ell_k}}^{Y}$\\($b$-agent)};
  \fill[red!70!white] (BB) circle (3pt) node[above, yshift=20pt, text=black] {$B_{j,{\ell_k}}^{Y}$\\($b$-agent)};

  \draw[thick, fill=white] (PortB) circle (3pt) node[right, xshift=4pt] {Port Vertex $p_{Y,j}$\\ (Hole)};

\end{tikzpicture}
\caption{Constraint transmission along an orthogonally routed link edge wire. The physical occupation of the port vertex $i$ consumes the hole at $B_{i,{\ell_k}}^{X}$, rigidly forcing $b$-agents to transition along the orthogonal corners. This guarantees $B_{j,{\ell_k}}^{Y}$ is occupied, which preserves the port vertex $p_{Y,j}$ as a hole.}
\label{fig:link-wire}
\end{figure}
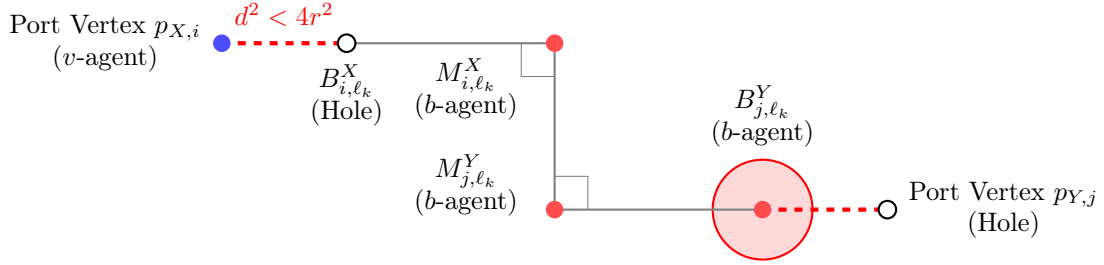

\subsection{Start and goal configurations}

To complete the construction of the instance $I_{\text{LA}}$ of LA-MAPF, we must define the start configuration $C_s$ and the goal configuration $C_g$ based on the standard token configurations $T_s$ and $T_g$.

\paragraph{Token Triangles ($v$-agents):} For every token triangle $X$ of $G_{\st}$, the unique $v$-agent assigned to $X$ is placed at the port vertex $p_{X,i}$ corresponding to the port $i$ occupied by the token in $T_s$ (specifically, if $i=2$, we place the agent at $p_{X,2}$, leaving the vertex $p_{X,t}$ unoccupied). Its goal position is the port vertex $p_{X,j}$ determined in the same manner based on the port $j$ occupied in $T_g$.

\paragraph{Token Edges ($v$-agents):} 
For every token edge $Y$ with the ports $1$ and $2$, the sequence of $v$-agents leaves exactly one unoccupied vertex as a single hole. In the start configuration $C_s$, if the token is at the port $1$ in $T_s$, the hole is placed at the port vertex $p_{Y,2}$ corresponding to the opposite port $2$, which forces the $v$-agents to occupy the intermediate orthogonal corners and the port vertex $p_{Y,1}$. Once this initial setup is established, because the physical narrowness of the wire prevents agents from bypassing one another, their relative spatial order becomes permanently locked. The goal configuration $C_g$ assigns the hole based on $T_g$ in the same manner. Since the relative order of the agents remains exactly as it was set in $C_s$, simply transitioning the hole to its goal position naturally and fully determines the final goal positions of all $v$-agents along the wire.

\paragraph{Link Edges ($b$-agents):} 
For the orthogonal wire $W_{\ell_k}$ corresponding to each link edge ${\ell_k}$ joining the port $i$ of one component $X$ to the port $j$ of another component $Y$, the assigned $b$-agents between its blocking vertices $B_{i,{\ell_k}}^{X}$ and $B_{j,{\ell_k}}^{Y}$ are similarly configured to leave exactly one unoccupied vertex as a single hole. In the start configuration $C_s$, if the start token configuration $T_s$ has a token at the port $i$ of $X$, the associated port vertex must be occupied by a $v$-agent. Due to the physical disk overlap, this $v$-agent blocks $B_{i,{\ell_k}}^{X}$ from hosting any $b$-agent, forcing the hole to start at $B_{i,{\ell_k}}^{X}$. If neither port connected by $\ell_k$ is occupied in $T_s$, the hole can be placed arbitrarily at either $B_{i,\ell_k}^X$ or $B_{j,\ell_k}^Y$. The remaining $b$-agents simply fill the rest of the wire. Just as with the token edges, once these agents are placed in $C_s$, their inability to bypass one another guarantees that their initial relative sequence is preserved until the end. The goal configuration $C_g$ are constructed identically based on $T_g$, and this preserved order ensures that dictating the hole's final location uniquely specifies the exact goal positions of all $b$-agents.

This completes the construction of the instance $I_{\text{LA}}$ of LA-MAPF (Figure \ref{fig:reduction-pipeline}).

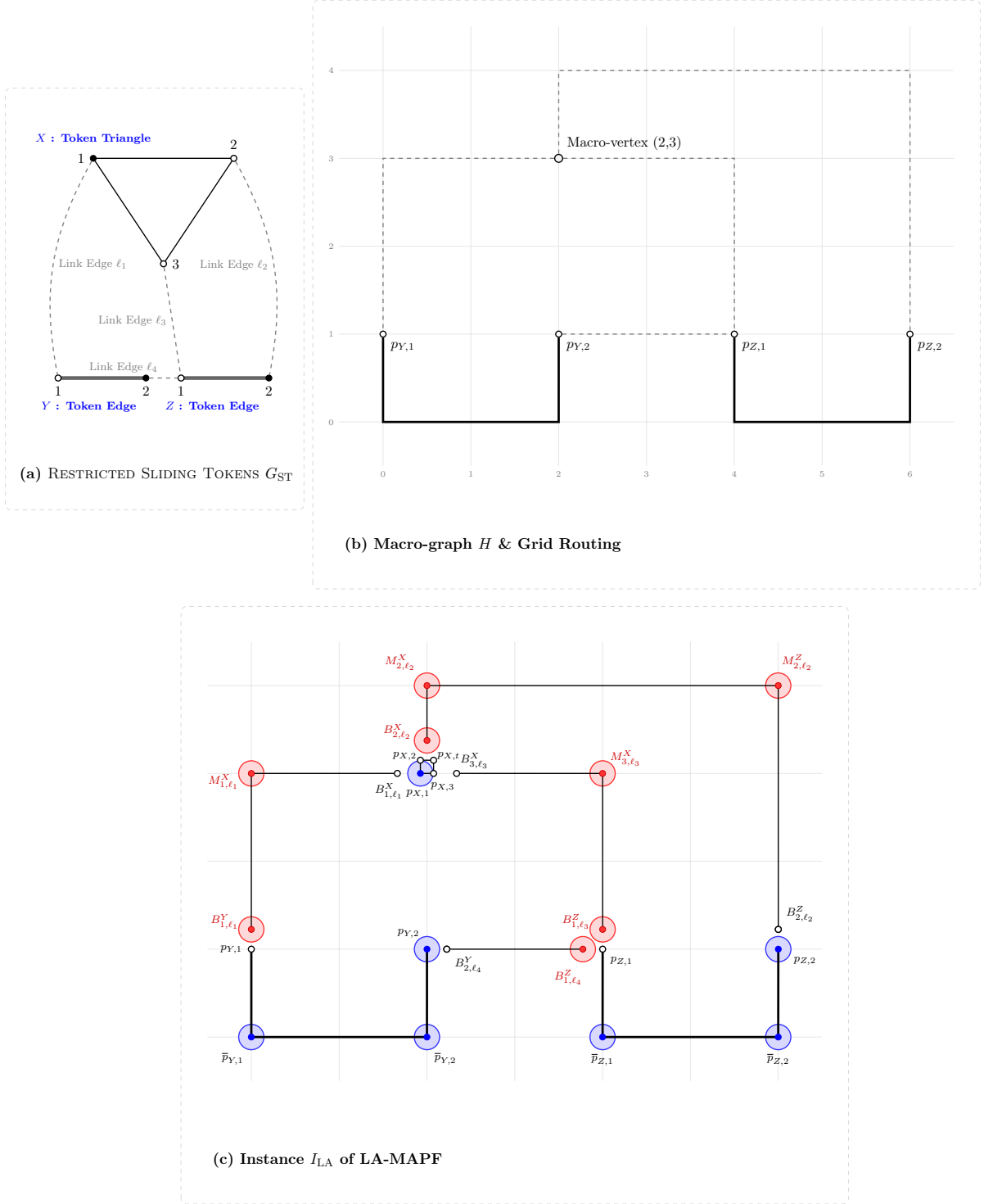
\begin{figure}[htbp]
\centering
\resizebox{\textwidth}{!}{
\begin{tikzpicture}[
    scale=2.2, 
    every node/.style={font=\small},
    token/.style={circle, draw=black, fill=black, inner sep=1.5pt},
    empty/.style={circle, draw=black, fill=white, thick, inner sep=1.5pt},
    vagentbg/.style={circle, draw=blue!80, fill=blue!15, thick, minimum size=18pt, inner sep=0pt}, 
    vagent/.style={circle, draw=blue, fill=blue, inner sep=1.5pt}, 
    bagentbg/.style={circle, draw=red!80, fill=red!15, thick, minimum size=18pt, inner sep=0pt},
    bagent/.style={circle, draw=red!80!black, fill=red!80, inner sep=1.5pt}, 
    hole/.style={circle, draw=black, fill=white, thick, inner sep=1.5pt} 
]

\begin{scope}[shift={(0, 3)}]

    \node[font=\bfseries, right] at (-0.7, -0.6) {(a) {\sc Restricted {\sc Sliding Tokens}} $G_{\rm ST}$};

    \coordinate (X1) at (0.2, 3);
    \coordinate (X2) at (1.8, 3);
    \coordinate (X3) at (1.0, 1.8);

    \draw[thick] (X1) -- (X2) -- (X3) -- cycle;

    \node[token, label=left:$1$] (nX1) at (X1) {};
    \node[empty, label=above:$2$] (nX2) at (X2) {};
    \node[empty, label=right:$3$] (nX3) at (X3) {};

    \node[above=0.15cm of nX1, text=blue, font=\scriptsize\bfseries] {$X$ : Token Triangle};

    \coordinate (Y1) at (-0.2, 0.5);
    \coordinate (Y2) at (0.8, 0.5);

    \draw[thick, double] (Y1) -- (Y2);

    \node[empty, label=below:$1$] (nY1) at (Y1) {};
    \node[token, label=below:$2$] (nY2) at (Y2) {};
    \node[below left=0.4cm and 0.05cm of nY2, text=blue, font=\scriptsize\bfseries] {$Y$ : Token Edge};

    \coordinate (Z1) at (1.2, 0.5);
    \coordinate (Z2) at (2.2, 0.5);

    \draw[thick, double] (Z1) -- (Z2);

    \node[empty, label=below:$1$] (nZ1) at (Z1) {};
    \node[token, label=below:$2$] (nZ2) at (Z2) {};
    \node[below left=0.4cm and 0.05cm of nZ2, text=blue, font=\scriptsize\bfseries] {$Z$ : Token Edge};

    \draw[dashed, thick, gray] (nX1) edge[bend right=20] node[right, font=\scriptsize] {Link Edge $\ell_1$} (nY1);
    \draw[dashed, thick, gray] (nX2) edge[bend left=20] node[left, font=\scriptsize] {Link Edge $\ell_2$} (nZ2);
    \draw[dashed, thick, gray] (nX3) -- node[left, font=\scriptsize] {Link Edge $\ell_3$} (nZ1);
    \draw[dashed, thick, gray] (nY2) -- node[above left, font=\scriptsize] {Link Edge $\ell_4$} (nZ1);

    \draw[gray!30, dashed, rounded corners] (-0.8, -1.0) rectangle (2.6, 3.8);

\end{scope}

\begin{scope}[shift={(3.5, 3)}]

    \node[font=\bfseries, right] at (-0.5, -1.4) {(b) Macro-graph $H$ \& Grid Routing};
    
    \draw[step=1, gray!20, thin] (-0.5, -0.5) grid (6.5, 4.5);
    \foreach \x in {0,1,2,3,4,5,6} \node[font=\tiny, gray, below] at (\x, -0.5) {\x};
    \foreach \y in {0,1,2,3,4} \node[font=\tiny, gray, left] at (-0.5, \y) {\y};

    \node[empty, label=below right:$p_{Y,1}$] (gY1) at (0, 1) {};
    \node[empty, label=below right:$p_{Y,2}$] (gY2) at (2, 1) {};
    \node[empty, label=below right:$p_{Z,1}$] (gZ1) at (4, 1) {};
    \node[empty, label=below right:$p_{Z,2}$] (gZ2) at (6, 1) {};
    
    \node[draw, thick, fill=gray!10, circle, inner sep=2pt, label=above right:{Macro-vertex (2,3)}] (gX) at (2, 3) {};

    \draw[ultra thick, black] (gY1) -- (0, 0) -- (2, 0) -- (gY2);
    \draw[ultra thick, black] (gZ1) -- (4, 0) -- (6, 0) -- (gZ2);

    \draw[dashed, thick, gray] (gX) -- (0, 3) -- (gY1); 
    \draw[dashed, thick, gray] (gX) -- (2, 4) -- (6, 4) -- (gZ2); 
    \draw[dashed, thick, gray] (gX) -- (4, 3) -- (gZ1); 
    \draw[dashed, thick, gray] (gY2) -- (gZ1);

    \draw[gray!30, dashed, rounded corners] (-0.8, -1.9) rectangle (6.8, 4.8);
\end{scope}

\begin{scope}[shift={(2, -4)}]
    
    \node[font=\bfseries, right] at (-0.5, -1.4) {(c)  Instance $I_{\mathrm{LA}}$ of LA-MAPF};
    
    \draw[step=1, gray!20, thin] (-0.5, -0.5) grid (6.5, 4.5);
    
    \coordinate (GridPoint) at (2,3);
    \node[font=\scriptsize, gray, above left] at (2.0, 3.0) {};
    
    \coordinate (P1) at (1.925, 3.0);
    \coordinate (P2) at (1.925, 3.15);
    \coordinate (Pt) at (2.075, 3.15);
    \coordinate (P3) at (2.075, 3.0);
    
    \coordinate (B1) at (1.6625, 3.0);
    \coordinate (B2) at (2.0, 3.375);
    \coordinate (B3) at (2.3375, 3.0);
    
    \coordinate (Y1) at (0,1); \coordinate (B_Y1_N) at (0, 1.225);
    \coordinate (Y2) at (2,1); \coordinate (B_Y2_E) at (2.225, 1);
    \coordinate (Z1) at (4,1); \coordinate (B_Z1_N) at (4, 1.225); \coordinate (B_Z1_W) at (3.775, 1);
    \coordinate (Z2) at (6,1); \coordinate (B_Z2_N) at (6, 1.225);

    \node[vagentbg] at (P1) {};
    
    \node[vagentbg] at (0,0) {}; 
    \node[vagentbg] at (2,0) {};
    \node[vagentbg] at (Y2) {}; 
    \node[vagentbg] at (4,0) {}; 
    \node[vagentbg] at (6,0) {};
    \node[vagentbg] at (Z2) {};
    
    \node[bagentbg] at (0,3) {};
    \node[bagentbg] at (B_Y1_N) {};
    \node[bagentbg] at (B2) {};
    \node[bagentbg] at (2,4) {};
    \node[bagentbg] at (6,4) {};
    \node[bagentbg] at (4,3) {};
    \node[bagentbg] at (B_Z1_N) {};
    \node[bagentbg] at (B_Z1_W) {};

    \draw[ultra thick, black] (Y1) -- (0,0) -- (2,0) -- (Y2);
    \draw[ultra thick, black] (Z1) -- (4,0) -- (6,0) -- (Z2);
    
    \draw[thick, black] (B1) -- (0,3) -- (B_Y1_N);
    \draw[thick, black] (B2) -- (2,4) -- (6,4) -- (B_Z2_N);
    \draw[thick, black] (B3) -- (4,3) -- (B_Z1_N);
    \draw[thick, black] (B_Y2_E) -- (B_Z1_W);

    \draw[thick] (P1) -- (P2) -- (Pt) -- (P3) -- cycle;

    \node[vagent] at (P1) {};
    \node[empty] at (P2) {};
    \node[empty] at (Pt) {};
    \node[empty] at (P3) {};
    
    \node[hole] at (B1) {};
    \node[bagent] at (B2) {};
    \node[hole] at (B3) {};

    \node[font=\scriptsize, text=black] at (1.9, 2.77) {$p_{X,1}$};
    \node[font=\scriptsize, text=black] at (1.75, 3.2) {$p_{X,2}$};
    \node[font=\scriptsize, text=black] at (2.25, 3.2) {$p_{X,t}$};
    \node[font=\scriptsize, text=black] at (2.18, 2.85) {$p_{X,3}$};
    
    \node[font=\scriptsize, xshift=-2mm, yshift=0mm, text=black] at (1.66, 2.80) {$B_{1,\ell_1}^{X}$};
    \node[font=\scriptsize, xshift=-3mm, yshift=0mm,text=red!80!black] at (1.80, 3.48) {$B_{2,\ell_2}^{X}$};
    \node[font=\scriptsize, xshift=2mm, yshift=0mm, text=black] at (2.45, 3.15) {$B_{3,\ell_3}^{X}$};

    \node[hole, label={[font=\scriptsize]left:$p_{Y,1}$}] at (Y1) {}; 
    \node[vagent, label={[font=\scriptsize, xshift=0mm, yshift=-2mm]below left:$\overline{p}_{Y,1}$}] at (0,0) {}; 
    \node[vagent, label={[font=\scriptsize, xshift=0mm, yshift=-2mm]below right:$\overline{p}_{Y,2}$}] at (2,0) {};
    \node[vagent, label={[font=\scriptsize, xshift=0mm, yshift=1mm]above left:$p_{Y,2}$}] at (Y2) {}; 
    \node[hole] at (B_Y2_E) {};

    \node[hole, label={[font=\scriptsize]below right:$p_{Z,1}$}] at (Z1) {}; 
    \node[vagent, label={[font=\scriptsize, xshift=0mm, yshift=-2mm]below:$\overline{p}_{Z,1}$}] at (4,0) {}; 
    \node[vagent, label={[font=\scriptsize, xshift=0mm, yshift=-2mm]below:$\overline{p}_{Z,2}$}] at (6,0) {};
    \node[vagent, label={[font=\scriptsize, xshift=2mm, yshift=0mm]below right:$p_{Z,2}$}] at (Z2) {};
    \node[hole] at (B_Z2_N) {};

    \node[bagent, label={[font=\scriptsize, xshift=-1mm, yshift=-2mm,text=red!80!black]left:${M}_{1,\ell_1}^{X}$}] at (0,3) {};
    \node[bagent, label={[font=\scriptsize, xshift=-1mm, yshift=2mm,text=red!80!black]left:$B_{1,\ell_1}^{Y}$}] at (B_Y1_N) {};

    \node[bagent, label={[font=\scriptsize, , xshift=-1mm, yshift=2mm, text=red!80!black]above left:${M}_{2,\ell_2}^{X}$}] at (2,4) {};
    \node[bagent, label={[font=\scriptsize, , xshift=-1mm, yshift=2mm, text=red!80!black]above right:${M}_{2,\ell_2}^{Z}$}] at (6,4) {};
    \node[hole, label={[font=\scriptsize]above right:$B_{2,\ell_2}^{Z}$}] at (B_Z2_N) {};

    \node[bagent, label={[font=\scriptsize, text=red!80!black]above right:${M}_{3,\ell_3}^{X}$}] at (4,3) {};
    \node[bagent, label={[font=\scriptsize, xshift=-1mm, yshift=2mm,text=red!80!black]left:$B_{1,\ell_3}^{Z}$}] at (B_Z1_N) {};

    \node[hole, label={[font=\scriptsize]below right:$B_{2,\ell_4}^{Y}$}] at (B_Y2_E) {}; 
    
    \node[bagent, label={[font=\scriptsize, xshift=2mm, yshift=-7mm,text=red!80!black]left:$B_{1,\ell_4}^{Z}$}] at (B_Z1_W) {};

    \draw[gray!30, dashed, rounded corners] (-0.8, -1.9) rectangle (6.8, 4.9);
\end{scope}
\end{tikzpicture}
}
\caption{A comprehensive macro-to-micro geometric reduction pipeline. \textbf{(a)} A standard configuration in the instance ($G_{\rm ST}, T_s, T_g$) of {\sc Restricted {\sc Sliding Tokens}}, where black circles represent tokens and white circles are empty. \textbf{(b)} The macro-graph $H$ embedded on an orthogonal grid, where the token triangle $X$ acts as a macro-vertex assigned to a single grid intersection (e.g., the coordinate $(2,3)$). \textbf{(c)} The resulting scaled LA-MAPF Euclidean workspace. Both $v$-agents (blue) and $b$-agents (red) are represented with their physical collision disks of radius $r$. Because the physical disks of the agents cannot bypass each other in the orthogonal wires, their relative order is fixed.}
\label{fig:reduction-pipeline}
\end{figure}

\subsection{Polynomial size}

\begin{lemma}
The construction of $I_{\mathrm{LA}}$ can be carried out in polynomial time and the size of $I_{\mathrm{LA}}$ is bounded by a polynomial in the size of $G_{\st}$.
\end{lemma}

\begin{proof}
The macro-graph $H=(V_H, E_H)$ contains one macro-vertex per token triangle and two port vertices per token edge. Thus, $|V_H| \le |V(G_{\st})|$. Recall that the orthogonal embedding can be computed in linear time and can be drawn on a $|V_H| \times |V_H|$ grid~\cite{BiedlK98}. Thus, the coordinates of all grid points and orthogonal bends can be represented using at most $O(\log |V_H|)$ bits. Replacing macro-vertices with constant-size triangle gadgets, and routing wires by placing vertices at orthogonal bends, generates a total number of vertices, edges, and agents that is polynomial with respect to $|V(G_{\st})|$.

The exact geometric coordinates are computed by mapping integer grid points to $\mathbb{Q}^2$, scaled by a constant factor $\Lambda = 100r$, and applying specific rational offsets (e.g., $1.75r$ and $1.5r$) for local blocking vertices. Therefore, all coordinates possess a polynomial binary encoding length.
\end{proof}

\subsection{Correctness}

We now prove the correctness of the reduction. 

\begin{lemma}
If the instance $I_{\la}$ of  \emph{LA-MAPF} has a conflict-free plan from $C_s$ to $C_g$, then in the original instance $(G_{\rm ST}, T_s, T_g)$ of {\sc Restricted {\sc Sliding Tokens}}, there exists a sequence of valid token moves that transforms $T_s$ into $T_g$.
\end{lemma}

\begin{proof}
Let $C_0 = C_s, C_1, C_2, \dots, C_k = C_g$ be a conflict-free plan in $I_{\text{LA}}$. We construct a sequence of valid token moves by defining a mapping from each configuration $C_i$ to a token configuration $T_i$ of $G_{\rm ST}$. For each $C_i$, the token configuration $T_i$ is determined as follows:

\begin{itemize}
    \item \textbf{For each token triangle $X$:} 
    If the $v$-agent is at the port vertex $p_{X,1}$, place the token at the port 1 in $T_i$. 
    If the $v$-agent is at $p_{X,3}$, place the token at the port 3 in $T_i$. 
    If the $v$-agent is at $p_{X,2}$ or $p_{X,t}$, place the token at the port 2 in $T_i$.
    
    \item \textbf{For each token edge $Y$:}
    If the unique hole in the token edge wire is at $p_{Y,1}$, place the token at the port 2 in $T_i$. 
    If the hole is at $p_{Y,2}$, place the token at the port 1 in $T_i$. 
    If both $p_{Y,1}$ and $p_{Y,2}$ are occupied by $v$-agents, place the token at the port 1 in $T_i$.
\end{itemize}

First, we prove that for any $i$, $T_i$ is a token configuration in $G_{\rm ST}$. 
According to our mapping rules, a token is assigned to a port only if its corresponding port vertex is occupied by a $v$-agent in $C_i$. By Lemma 4, for any link edge $\ell_m$ connecting the port $n$ of a component $X$ and the port $w$ of a component $Y$, the two corresponding port vertices $p_{X,n}$ and $p_{Y,w}$ cannot be occupied by $v$-agents in $C_i$. Therefore, in the mapped configuration $T_i$, no link edge can have tokens at both of its endpoints. Thus, $T_i$ satisfies the independent set constraint.

Second, we prove that the mapped token configurations $T_i$ and $T_{i+1}$ differ by at most one token position. 
The transition $C_i \rightsquigarrow C_{i+1}$ involves one agent moving to an adjacent vertex. 
If a $b$-agent moves, the positions of all $v$-agents and holes within the token triangle gadgets and token edge wires remain unchanged, yielding $T_{i+1} = T_i$. 
If a $v$-agent within the gadget corresponding to a token triangle $X$ moves, only the token corresponding to $X$ will alter its mapped port. 
If a $v$-agent within the orthogonal wire representing a token edge $Y$ moves, the unique hole in this wire cascades by one step. 
According to our mapping for token edges, the token's position in $T$ changes between the port 1 and the port 2 if and only if the position of the hole moves to or from the port vertex $p_{Y,1}$. In all cases, the transition $C_i \rightsquigarrow C_{i+1}$ causes at most one token to change its position in the mapped token configurations.

Finally, by removing $T_{i+1}$ from the sequence $T_0, T_1, \dots, T_k$ whenever $T_i = T_{i+1}$, we obtain a sequence $T'_0, T'_1, \dots, T'_q$ satisfying $T'_j \neq T'_{j+1}$ for all $j$. Because every step between consecutive token configurations $T'_j$ and $T'_{j+1}$ involves one token changing its position, it constitutes a valid token move in $G_{\rm ST}$. Since the start configuration $C_s = C_0$ maps to $T'_0 = T_s$ and the goal configuration $C_g = C_k$ maps to $T'_q = T_g$, the sequence $T'_0, \dots, T'_q$ demonstrates a sequence of valid token moves that transforms $T_s$ into $T_g$.
\end{proof}

\begin{lemma}
If the instance $(G_{\rm ST},T_s,T_g)$ of {\sc Restricted {\sc Sliding Tokens}} has a sequence of token moves that transforms $T_s$ to $T_g$, then the instance $I_{\la}$ of  \emph{LA-MAPF} has a conflict-free plan from the start configuration $C_s$ to the goal configuration $C_g$.
\end{lemma}

\begin{proof}
Let $T_s = T_0, T_1, \dots, T_q = T_g$ be a sequence of token configurations such that $T_t$ can be transformed to $T_{t+1}$ via a valid token move for all $t$. We simulate each valid token move that transforms $T_t$ to $T_{t+1}$ using a finite sequence of collision-free transitions.

Suppose the valid token move slides the token within a component $X$ from the port $i$ to an adjacent port $k$. To simulate this valid token move in $I_{\la}$, we must first vacate every blocking vertex associated with the target port vertex $p_{X,k}$. Consider a link edge $\ell_m$ incident to the port $k$, connecting to the port $j$ of a component $Y$. Because $T_{t+1}$ is a token configuration, the port $j$ of the component $Y$ is currently unoccupied by a token (meaning no $v$-agent is at the port vertex $p_{Y,j}$). Thus, $B_{j,\ell_m}^Y$ is structurally safe. We can move the $b$-agents along $W_{\ell_m}$ one by one, transferring the hole to $B_{k,\ell_m}^X$. 

Once all blocking vertices incident to $p_{X,k}$ are holes, the $v$-agent within $X$ can move to $p_{X,k}$. We can prove that the resulting configuration is mapped to $T_{t+1}$ in the same way as the proof of Lemma 6. By repeating this procedure, we route all $v$-agents to their respective goal positions defined in the goal configuration $C_g$. Finally, for any link edge wire whose hole does not yet reach its goal position in $C_g$, we simply move its $b$-agents to reposition the hole. This guarantees that $I_{\text{LA}}$ has a conflict-free plan from $C_s$ to $C_g$. 
\end{proof}

\begin{lemma}
$\mathrm{LA\text{-}MAPF}$ is $\mathrm{PSPACE\text{-}hard}$.   
\end{lemma}

\begin{proof}
Our reduction maps any instance of {\sc Restricted {\sc Sliding Tokens}} to an instance of LA-MAPF in polynomial time. By Lemmas 6 and 7, the two instances share the same answer. Since {\sc Restricted {\sc Sliding Tokens}} is PSPACE-complete, LA-MAPF is PSPACE-hard.
\end{proof}

Our reduction demonstrates that LA-MAPF remains PSPACE-hard even if the agents' movements are restricted to horizontal and vertical directions, and even if the underlying graph of the given instance of LA-MAPF  is a plane graph. Together with membership in PSPACE (Lemma 3), Lemma 8 proves Theorem 1 and Corollary 2.

\section{Conclusion}
We established that the computational complexity of Multi-Agent Path Finding for Large Agents (LA-MAPF) is PSPACE-complete. More concretely, its PSPACE-hardness is proven via a geometric reduction from {\sc Restricted {\sc Sliding Tokens}}.

\bibliographystyle{plain}
\bibliography{complexity_lamapf_bib}

\end{document}